\documentclass[10pt,journal,cspaper,compsoc]{IEEEtran}
\usepackage[ruled,vlined]{algorithm2e}
\usepackage{flushend}
\newtheorem{thm}{Theorem}
\newtheorem{lemma}{Lemma}
\usepackage{times,amsmath,epsfig}
\usepackage{color}

\begin{document}

\title{Heuristic-based Optimal Resource Provisioning in Application-centric Cloud}

\author{Sunirmal~Khatua,~\IEEEmembership{Member,~IEEE,}
        Preetam~K.~Sur,
        Rajib~K.~Das~ 
        and Nandini~Mukherjee,~\IEEEmembership{Senior~Member,~IEEE}

\IEEEcompsocitemizethanks{\IEEEcompsocthanksitem S. Khatua and R.K.Das  are with the Department
of Computer Science and Engineering, University of Calcutta, India.\protect\\
E-mail: skhatuacomp@caluniv.ac.in and rajib.das.k@gmail.com
\IEEEcompsocthanksitem P. K. Sur is with the Jadavpur University, Kolkata, India.\protect\\
E-mail: mailtopreetam@gmail.com
\IEEEcompsocthanksitem Nandini Mukherjee is with the Department
of Computer Science and Engineering,Jadavpur University, India\protect\\
E-mail:nmukherjee@cse.jdvu.ac.in}
\thanks{}}

\markboth{IEEE Transactions on Cloud Computing, 2014}%
{Shell \MakeLowercase{\textit{et al.}}: Bare Demo of IEEEtran.cls for Computer Society Journals}

\IEEEcompsoctitleabstractindextext{%
\begin{abstract}
Cloud Service Providers~(CSPs) adapt different pricing  models
for their offered services. Some of the models are suitable for
short term requirement while others may be suitable for the Cloud
Service User's~(CSU) long term requirement. In this paper, we
look at the problem of finding the amount of resources to be
reserved to satisfy the CSU's long term demands with the
aim of minimizing the total cost. Finding the optimal resource
requirement to satisfy the the CSU's demand for resources needs
sufficient research effort. Various algorithms were discussed in
the last couple of years for finding the optimal resource
requirement but most of them are based on IPP which is NP in nature. In this paper, we
derive some heuristic-based polynomial time algorithms to find
some near optimal solution to the problem. We show that the cost
for CSU using our approach is comparable to the solution obtained
using optimal Integer Programming Problem(IPP).

\end{abstract}
}

\maketitle

\IEEEdisplaynotcompsoctitleabstractindextext
\IEEEpeerreviewmaketitle

\section{Introduction}

Recent advances in distributed computing has brought a paradigm
shift in the computing world. The practices of dedicated access
to computers owned by individuals or organizations have been
replaced by on-demand accesses to resources shared among many
individuals and different organizations. Cloud Computing has
largely contributed to this transition by enabling ubiquitous,
convenient, on-demand network access to a shared pool of
configurable computing resources (e.g., networks, servers,
storage, applications, and services) that can be rapidly
provisioned and released with minimal management effort or
service provider interaction \cite{buyaa08}

Present practice in the Cloud Computing arena is that the Cloud
Service Users (CSUs) request for resources and Cloud Service
Providers (CSPs) allot the virtualised resources to the CSUs
according to their requirements. However, while requesting the
resources, CSUs face a major challenge due to various pricing
schemes offered by the CSPs. Resources are available on
reservation basis, as well as on demand basis. Reservation is to
be done for a fixed contract period with a fixed price. However,
reservation for a longer period or for larger amount of resources
than that required to meet the demand for
resources of an application, may lead to over-provisioning and
higher cost for the CSUs. On the other hand, on-demand prices are
generally higher than the reserved one and therefore if the resources are allotted only
on-demand basis, cost to be paid by the CSUs will again be high.
Therefore, some optimization strategies are required to reduce
the resource usage charges from the users' point of view.

This paper focuses on the optimizing strategies aimed at
reducing the total cost of cloud deployment. Remaining part of
the paper is organised as follows. Various pricing schemes offered by different CSPs are discussed in the next section. A brief overview of the related work is given in section \ref{relatedwork}. An IPP based optimal resource provisioning strategy for the application-centric cloud framework is discussed in Section~ \ref{sec_ipp}. Section \ref{sec_heuristic} and \ref{sec:heuristic1} present the proposed heuristics for polynomial time resource reservation strategies. The performance of the heuristic based approach compared to the IPP is presented in Section \ref{sec:result} and section \ref{sec:conclude} concludes the paper.

\section{Pricing Schemes}
\label{price} Cloud computing service providers offer computing
resources as utility and software as a service over a network.
Cloud users pay for these resources or services on the basis of
their usage. Optimising the cost of resources and services from
the {\em Cloud Service User's} point of view is a tricky issue as
the CSPs often provide non-fixed pricing models for their
services. A brief description of the general trend in the pricing
models offered by the CSPs is given below.

\begin{enumerate}
\item Fixed Costs - Resources are charged by its type and duration ( e.g. month, year etc).
Here exactly one price is assigned for a fixed time duration
(say, for a month). The cost is computed as the fixed charge
multiplied by the number of resource instances or service units
requested by the consumer. The users pay the fixed cost even
though the resources may not be used for the entire time
duration.
\item Variable Costs - Under variable cost model, Cloud computing delivers computing power and
services to consumers on a variable cost pay-as-you-go basis
determined by the number of users and their volume of
transactions. Resources are charged by their types and usage
(e.g. per hour) and with no long-term commitments or upfront
payments. Here the CSUs need to pay on different items and deal
with several different structures. Resources are allotted
on-demand basis and the users do not need to pay if the resources
are not used.

\item Hybrid Costs - Hybrid Costs are a mixed form of the former two cost structures, where one can find a variable part and a fixed part.

\item Flexible Costs - Resources are charged by its type
and time of usage. At a specific instance of time, the cost of
the resource is set depending on the demand of the resources.
Like variable costs, users do not need to pay if they do not use
the resources.
\end{enumerate}

In this paper, we have used the pricing schemes offered by Amazon
as a case study. Amazon EC2, one of the major IaaS
(Infrastructure As A Service) provider, offer pricing schemes in
the form of On-Demand, Reserved and Spot instances. On-Demand
instances support variable costs having a fixed hourly charge of
usage. Reserved instances support the hybrid cost model with a
fixed one time reservation cost, which depends on the contract
duration (e.g. 1 year or 3 years) and a hourly usage cost which
is a discounted price over the cost of the On-Demand instances.
Flexible costs are supported by the Spot Instances which are
similar to On-Demand but having flexible hourly charge of usage.
So cost optimization can be considered on Spot instances,
Reserved instances, as well as on both. In this paper, we have not
considered Spot instances for cost minimization.

\section{Related Work}\label{relatedwork}
Various aspects of multiple pricing schemes have
been studied in \cite{pricing}-\cite{optimization3}. Most of these
research works consider reserved versus on-demand pricing scheme. A model for
determining the price of a reserved instance as well as an
on-demand instance is proposed in \cite{pricing}. The research work
presented in \cite{reserved_vs_ondemand} considers the revenue maximization
problem from a CSP's point of view, separating the CSUs into
`premium'~(with reserved services) and `basic'~(with on-demand
services) users. Both papers demonstrate the effectiveness of having
multiple pricing schemes, such as reservation and on-demand which are
followed by most of the CSPs now-a-days.

Cost optimization from a CSU's point of view, by
considering different pricing schemes, have been studied in
\cite{optimization1},  \cite{optimization2} and \cite{optimization3}.
A stochastic integer
programming model has been adopted in \cite{optimization1} to
optimize the cost of SLA-aware resource scheduling in cloud.
H.~Menglan et al., in their paper \cite{optimization2}, consider
on-demand and reserved instances to optimize cost for deadline
constrained jobs and to optimize total processing time for budget
constrained jobs. An optimal cloud resource provisioning (OCRP)
algorithm is proposed by S.~Chaisiri et al. in their paper
\cite{optimization3}. They have also considered a stochastic
integer programming model with demand and price uncertainty in
the OCRP algorithm. They have extended their work by considering
spot pricing scheme in \cite{optimization3}. These solutions need
the predicted demands for resources. Demand forecasting
models for twelve months have been developed in \cite{forecasting}
based on historical data.

Most of these previous works for optimizing cost use some form of
integer programming model which is NP in nature. They have
solved the problem using some techniques like benders decomposition
and sample-average approximation as have been used in \cite{optimization3}.
To the best of our knowledge, no one has used some good heuristics to
solve the cost optimization problem in polynomial time. Motivated
by these works, we have developed some heuristics to solve the cost
optimization problem in linear time. We have also
mathematically proved that our heuristic provides optimal
solution in certain scenarios. However, for simplicity, we do not
consider uncertainty of demands and prices in this work.

\section{Optimal Resource Provisioning in Application-Centric Cloud}\label{sec_ipp}
In our previous work \cite{aiccsa11}, we have proposed a resource
provisioning framework for application-centric cloud as shown in
Figure~\ref{fig_framework}.
\begin{figure}[h!]
\centering
\includegraphics[width=3.4in]{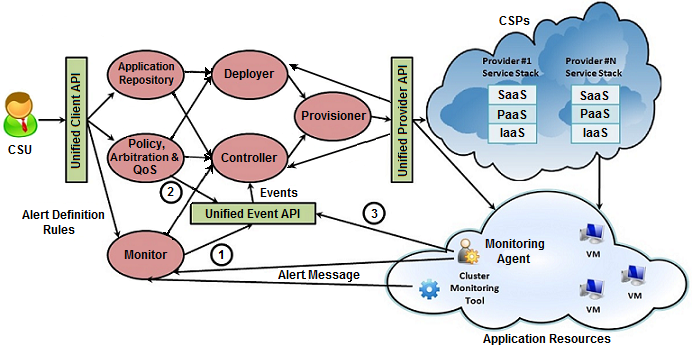}
\caption{Resource Provisioning Framework for Application-centric Cloud}
\label{fig_framework}
\end{figure}
Two key components of the framework are the Deployer module and
the Controller module. The Deployer module statically analyzes an
application to determine the optimal resource requirement for the
application. The Controller module dynamically
fine-tunes the provisioned resources to alleviate
the overprovisioning and underprovisioning situations.

In this paper, we extend the functionality of the Deployer module
to statically determine the amount of resources to be reserved in
advance in order to minimize the total cost of running an
application. While doing this, we make the following assumptions -
\begin{itemize}
\item An application is run through different stages denoted by $t, 1\le t\le T$. The granularity of a
stage is determined by the number of hours per stage ($h$).
\item At every stage of execution of the application, its demand for resources is known
(or estimated using the mechanisms described in \cite{forecasting}) and is
denoted by the vector $D$. Since the reservation made at any
particular time has a long term effect, in order to find an
optimal reservation schedule, the future value of demand for
resources is to be known in advance. It is therefore assumed that
the predicted values of the  demand vector are available at each
stage $t$, $1\le t\le T$.
\item Resources can be reserved under different contracts (e.g. for a 1 year contract or for a 3 year contract etc.)
\item Reserved resources have a one-time fixed charge for the duration of the contract and a variable usage
charge at discounted rate to be paid on hourly basis.
\item Resource instances can be reserved and launched at the beginning of a stage and can be terminated at the end of a stage.
If the demand at a given stage $t$ is more than the total amount of resources reserved, the balance requirement is satisfied with on-demand instances.
\end{itemize}

Thus, the optimal resource reservation problem can be formally stated as
\begin{quote}
\emph{Given a vector of demands D for resources at different
stages $t$, $1\le t\le T$, find the number of resources to
be reserved under different contracts  at
each stage $t$, so that the total cost is minimized satisfying the demands at every stage.}
\end{quote}

The above problem can be formulated as an Integer Programming Problem(IPP), which is discussed below.

\subsection{IPP-based Optimal Resource Reservation}
This section discusses an IPP formulation to find the optimal
resource reservation schedule to minimize cost. The various
notations used in the formulation are given in
Table~\ref{table_notations}. There are $K$ different types of
reservation contracts offered by a CSP. Each type
of contract~($k$) is associated with a one time reservation
cost~($R_k$), usage cost~($r_k$) per hour and its duration~($t_k$)
in number of stages. At every stage~($t$), we need to decide how
much instances to be reserved~($x_{t}^{R_k}$). We also need to determine the number instances to be launched based on reservation from contract $k$ ($x_{t}^{r_k}$) as well as on-demand~($x_{t}^{o}$). The cost of
on-demand instances is more than the usage cost of reserved
instances~($o > r_k$). At the same time, increasing reservation incurs high reservation cost. So, the objective is to find the optimal reservation balancing these two factors.

\begin{table}[h!]
\renewcommand{\arraystretch}{1.3}
\caption{Definitions and Notations}
\label{table_notations} \centering
\begin{tabular}{|l|l|}
\hline
\bfseries Symbol &\bfseries Definition \\
\hline\hline
$K$ &Number of reservation contracts offered by a CSP.\\
\hline
$T$ &Total number of stages for which the demands are available.\\
\hline
$h$ &Duration of each stage in hours\\
\hline
$D_t$ &The demand for instance-hours at stage $t$, $1 \le t \le T$.\\
\hline
$R_k$ &One time reservation charge per instance under contract $k$.\\
\hline
$t_k$ &Duration of contract $k$ in stages, $1\le k\le K$.\\
\hline
$r_k$ &Usage charge per hour for a reserved instance\\
 &under contract $k$.\\
\hline
$o$ &Usage charge per hour for an on-demand instance.\\
\hline
$x_{t}^{R_k}$ &Number of instances reserved under contract $k$ at stage $t$.\\
\hline
$x_{t}^{r_k}$ &Number of reserved instances launched under\\
& contract $k$ at stage $t$.\\
\hline
$x_{t}^{o}$ &Number of on−demand instances launched at stage $t$.\\
\hline
\end{tabular}
\end{table}

\subsection{IPP Formulation}
In this section, an integer programming problem is formulated with an objective of optimal resource provisioning.\\
Based on the previous definitions, the cost at any given stage $t$ is denoted by $C_t$, where\\\\
$C_t=\sum_{k=1}^K (x_{t}^{R_k}.R_k +x_t^{r_k}.r_k.h)+x_t^o.o.h$ :\\\\
The first term in the above expression stands for the cost of reservation under contract $k$, the second term stands for
the cost of using reserved instances, and the last term stands for the cost of using on-demand instances.

Our objective is to minimize the total cost over all stages with the following constraints -
\begin{enumerate}
  \item Number of reserved instances and number of instances launched from reserved instances as well as on-demand instances at any stage are non-negative integers.
  \item Instances launched at any stage $t$ under contract $k$ cannot exceed the total reserved instances under contract $k$. It is to be noted that
reservation under contract $k$ made at stage $t$ is effective up to stage $t+t_k-1$. In other words, only the instances which have been reserved at stage $t-t_k+1$ or after that are available at stage $t$.
  \item At any stage $t$, total number of instances launched under all the contracts plus on demand instances must satisfy the demand of the application.
\end{enumerate}

\noindent
Thus, the integer programming problem is as follows:

~\\
Minimize $\sum_{t=1}^T C_t$ subject to the following conditions  \\
\\
i) $x_{t}^{R_k}, x_{t}^{r_k}, x_{t}^{o} \ge 0~and~integer$\\
\\
ii) $ \sum_{i=t-t_k+1,i \ge 1}^{t} x_{i}^{R_k} \geq x_{t}^{r_k} ~~\forall~k=1,2,...,K ~\&~$\\$~~~~~~~~~~~~~~~~~~~~~~~~~~~~~~~~~~~~~~ \forall~t=1,2,...,T$ \\
\\
iii) $x_t^o + \sum_{k=1}^K x_{t}^{r_k}  \geq D_t ~~ \forall~t=1,2,...,T$
\\

\section{Heuristic-based Resource Reservation}\label{sec_heuristic}
An IPP is NP-hard and therefore, no polynomial time algorithm does
exist to solve the optimal resource reservation problem as
discussed in Section~\ref{sec_ipp}. In this
section, we derive some heuristics for solving the optimal
resource reservation problem in polynomial time. We show that
when there is a single type contract $k$ with contract duration
of $t_k$ stages, and the demand vector is available for stages
$t=1, \cdots ,t_k$, it is possible to determine the optimal value
for reservation under contract $k$. The discussion in this
section is based on the assumption that {\it duration of each stage is
1 hour.}

The usage cost for a demand $d$ in a single stage~(1 hour) with
$x$ reserved instances under contract $k$ is:

\begin{equation}\label{equ1}
\begin{array}{ccl}
\c{C}(x,d)& = & \left\{
\begin{array}{c l}
    d.r_k & if~x \geq d\\
    x.r_k + (d-x).o & otherwise
\end{array}\right. \\
& = & r_{k}.min(x,d) + o.max(d-x,0)
\end{array}
\end{equation}

Let us assume that the vector of demands $D$ for the duration of
$t_k$ stages is available. Let $x$ be the amount of resources
reserved under contract $k$ and let $E_x$ denotes the total cost
corresponding to demand vector $D$, that includes both, cost of
reservation and cost of using the resources. Thus, combining
equation~\ref{equ1} with the cost of reservation, $E_x$ can be
written as

\begin{equation}\label{equ_Ex}
E_x= x.R_k+ \sum_{i, D_i\le x} D_i.r_k + \sum_{i, D_i>x} (x.r_k+ (D_i-x).o)
\end{equation}

Let $D^s$ be the demand vector obtained by sorting $D$ in
ascending order. Considering $x=D_j^s$, Equ.~\ref{equ_Ex} can be
rewritten as

\begin{equation}
\label{cost vector j}
E_{D_j^s} =  D_{j}^{s}.R_{k}+ \sum_{i=1}^{j} D^s_i.r_{k} +\sum _{i=j+1}^{t_k}(D^s_j.r_k+(D_{i}^{s}-D_{j}^{s}).o)
\end{equation}

Now by replacing the on-demand cost $o$  with $r_k+\alpha_{k}$,
where $\alpha_k=o-r_k$ is the hourly discount for using reserved
instances under contract $k$ over on-demand instances, we have
\begin{equation}\label{eq2}
\begin{array}{lll}
E_{D_j^s} &=&D_{j}^{s}.R_{k}+\sum_{i=1}^{j}D_{i}^{s}.r_{k}+\\
 &&\sum_{i=j+1}^{t_{k}}\biggl[D_{j}^{s}.r_{k}+(D_{i}^{s}-D_{j}^{s}).(r_{k}+\alpha_{k})\biggr]\\
&=&D_{j}^{s}.R_{k} + \sum_{i=1}^{t_{k}}D_{i}^{s}.r_{k} + \sum_{i=j+1}^{t_{k}}(D_{i}^{s}-D_{j}^{s}).\alpha_{k}
\end{array}
\end{equation}

If the amount of reserved resources is increased to the next
higher value in the sorted demand vector,~(i.e. $D_{j+1}^S$ as
shown in Fig.~\ref{fig_demands}),then the total cost is given by
equation~\ref{cost vector j+1}.

\begin{equation}
\label{cost vector j+1}
E_{D_{j+1}^s} = D_{j+1}^{s}.R_{k} + \sum_{i=1}^{t_{k}}D_{i}^{s}.r_{k} + \sum_{i=j+2}^{t_{k}}(D_{i}^{s}-D_{j+1}^{s}).\alpha_{k}
\end{equation}

\begin{figure}[h!]
\centering
\includegraphics[width=3.0in]{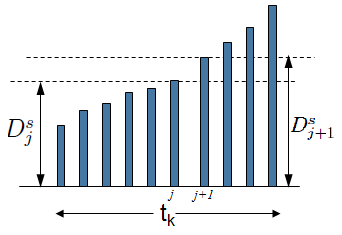}
\caption{Increasing reservation from $D_j^s$ to $D_{j+1}^s$.}
\label{fig_demands}
\end{figure}

Thus, the change in total expenditure due to change in reservation from $D_{j}^S$ to $D_{j+1}^S$ is

$\Delta{E}_{D_j^s,D_{j+1}^s}$
\begin{equation}
\label{cost change}
\begin{array}{ccl}
&=&(D_{j+1}^{s} - D_{j}^{s}).R_{k} - \sum_{i=j+2}^{t_{k}}((D_{j+1}^{s}-D_{j}^{s}).\alpha_{k} - \\
&&(D_{j+1}^{s}-D_{j}^{s}).\alpha_{k})\\
&=&(D_{j+1}^{s} - D_{j}^{s}).R_{k} - (D_{j+1}^{s}-D_{j}^{s}).\alpha_{k}(t_k - j)\\
&=&(D_{j+1}^{s} - D_{j}^{s})(R_{k} - t_{k}.\alpha_{k} + j.\alpha_{k})
\end{array}
\end{equation}

Now we state the following two lemmas based on the characteristics of reserved instances and equation \ref{cost change}.
\begin{lemma}\label{lemma_rk_tk}
For any contract type $k$, $R_k < t_k.\alpha_k$.
\end{lemma}

\begin{proof}
Consider a situation where demand is constant, i.e., $D_j=x$ for
$1\le j\le t_k$. If we reserve the amount $x$, the cost is $R_k.x
+ r_k.t_k.x$ and there is no on demand charge. If we make no
reservation then cost is $o.t_k.x$. The second option is
costlier, because if provisioning using only on-demand instances
is cheaper, then no one would go for reservation. Thus, $o.t_k.x>
R_k.x + r_k.t_k.x$, which implies $t_k.\alpha_k>R_k$, where
$o=r_k+\alpha_{k}$
\end{proof}
~
\begin{lemma}\label{lemma_member}
The number of instances to be reserved, for which the total cost
is minimized, is always a member of the demand vector.
\end{lemma}
\begin{proof}
We prove the above lemma by contradiction.  If possible,
let us consider the existence of $p$, $D_j^s < p < D_{j+1}^s$, such that
$E_{D_j^s} > E_p$ and $E_{D_{j+1}^s} > E_p$.

Now, we have, \\\\
$\begin{array}{ccl}
E_p & = & p.R_k + \sum_{i=1}^{j}{D_i^s.r_k}+ \\
&&\sum_{i=j+1}^{t_k}{[p.r_k + (D_i^s - p)(r_k + \alpha_k)]}$ and $\\\\
\Delta{E}_{D_j^s,p} & = & (p - D_j^s)(R_{k} - t_{k}.\alpha_{k} + j.\alpha_{k})\\
&=&(p - D_j^s).f(j,k)\\\\
\end{array}$
where, the expression $R_k - t_k.\alpha_k + j.\alpha_k$ is
denoted by $f(j,k)$. Similarly, it can be shown that\\\\
$\begin{array}{ccl}
\Delta{E}_{p,D_{j+1}^s}= (D_{j+1}^s - p).f(j,k)\\\\
\end{array}$
\\
As $D_j^s < p < D_{j+1}^s$, $\Delta{E}_{D_j^s,p}$ and
$\Delta{E}_{p,D_{j+1}^s}$ will have the same sign which will
depend on the sign of $f(j,k)$. But, this contradicts our initial
assumption that $E_{D_j^s} > E_p$ and $E_{D_{j+1}^s} > E_p$.
Hence, the proof.
\end{proof}
~

Based on lemma \ref{lemma_rk_tk} and lemma \ref{lemma_member},
the following theorem is stated.
\begin{thm}\label{t_minimum}
For a single type contract $k$, with demand vector $D$ available
for $t=1,2, \cdots t_k$ stages there exists a value of $j_k=
t_k-\lfloor R_k/\alpha_k\rfloor$, such that with reservation of
$D_{j_{k}}^s$, the cost is minimum.
\end{thm}
\begin{proof} By lemma \ref{lemma_rk_tk}, $R_k < t_k.\alpha_k$.
Earlier we assumed $f(j,k)=R_k-t_k.\alpha_k + j.\alpha_k$. Hence,
$f(j,k)$ has a negative value for $j < t_k-R_k/\alpha_k$ and a
positive value for $j>t_k-R_k/\alpha_k$.

As $\Delta E_{D^s_j,D^s_{j+1}}=(D_{j+1}^s-D_j^s)f(j,k)$ and
$D_{j+1}^s-D_j^s \ge 0$, we can write $E_{D_1^s} \ge E_{D_2^s}
\ge \cdots E_{D_{j_k}^s}$ and $E_{D_{j_k}^s} \le E_{D_{j_k+1}^s}
\le \cdots E_{D_{t_k}^s}$ where $j_k=t_k-\lfloor
R_k/\alpha_k\rfloor$.

The above result along with lemma \ref{lemma_member} proves that
cost is minimum when reservation is done for an amount of
resources equal to $D_{j_k}^s$ .
\end{proof}
~

It is to be noted that the index of the demand in the sorted
demand vector, for which the total cost is minimized, is
independent of the values of demand vector.

So, from theorem \ref{t_minimum} we
can conclude that for a single type contract $k$, the optimal
reservation is the $j^{th}$ smallest element in the demand vector
where $j = t_k - \lfloor R_k/\alpha_k\rfloor$. That means we do
not need to sort the demand vector since we can find the $j^{th}$
smallest element in $O(t_k)$ time using randomized select~\cite{cormen} algorithm.

In situations where demand vector is not available for entire duration $t_k$, but for a
smaller duration $t_k^\prime$, we just optimize the
cost over the duration $t_k^\prime$. Let $j=t_k^\prime -\lfloor R_k/\alpha_k\rfloor$. By using logic
similar to the above, it can be shown that optimal reservation in this case is $D_j^s$, if $j>0$ and it is zero,
otherwise ($j$ maybe negative for small values of $t_k^\prime$.  By lemma 1, $R_k<t_k.\alpha_k$ but
$R_k$ may be greater than $t_k^\prime.\alpha_k$.)

\section{Heuristic-based Resource Reservation Leading to Sub-Optimal Solutions}\label{sec:heuristic1}
The heuristic-based resource reservation as discussed in
Section~\ref{sec_heuristic} provides the optimal solution like
IPP-based resource reservation under the assumption of having a
single type contract $k$ and $|D| \leq t_k$. In this section we extend the heuristic to more general case
where these assumptions are relaxed to allow the following situations
\begin{enumerate}
  \item The demand vector size can be more than the contract duration~(i.e. $|D| > t_k$) and
  \item More than one contracts/providers are available~(i.e. considering multiple contracts).
\end{enumerate}
We propose the following algorithms to provide sub-optimal solutions for such situations.
First we consider single contract $k$ with any
demand duration and then we consider multiple contracts with any
demand duration.

\subsection{Single Contract Reservation}
In this reservation strategy, only one
contract~($k$) is considered at a time. The $k^{th}$ contract is
denoted by $C_k$ which is a tuple $(R_k,t_k,\alpha_k)$. The
demand duration may be more than $t_k$. The reservation decision
is taken for each contract duration~(called segment) separately
during the whole duration of the demand vector as shown in
Fig.~\ref{Fig_SingleContractReservation}.

Let the demand vector $D$ be available for stages
$t=1,2, \cdots T$. We divide the whole duration of $T$ into
multiple segments. If $T$ is divisible by $t_k$, then each
segment is of duration $t_k$.
 Otherwise, for   $i=1,2, \cdots \lfloor T/t_k\rfloor$ segment duration is $t_k$ and  the last segment is of duration
$T-t_k.\lfloor T/t_k\rfloor$. If we do not allow reservation made at two different stages to overlap,
 the optimal reservation for each segment is given by
theorem \ref{t_minimum} as shown in Section~\ref{sec_heuristic}
and the amount to be reserved at intervals of $t_k$ is given by
Algorithm~\ref{Algo_SingleContractReservationAlgo}.

\begin{figure}[h!]
\centering
\includegraphics[width=3.4in]{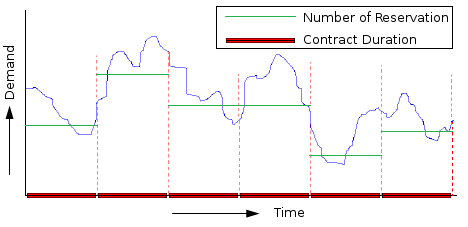}
\caption{Single Contract Reservation}
\label{Fig_SingleContractReservation}
\end{figure}

\begin{algorithm}
 \DontPrintSemicolon
 \SetKwInOut{Input}{input}\SetKwInOut{Output}{output}
 \Input{ Demand Vector $D[1..T]$, Contract $C_k$.}
 \Output{Reservation Decision in vector $x[1 ..\lceil \frac{T}{t_k} \rceil ]$}
 \Begin{
  $T_r \leftarrow T$\;
  $S \leftarrow 1$\;
  $i \leftarrow 1$\;
  \While {$T_r>0$}
  {$t_d \leftarrow \min (T_r,t_k$)\; /* $t_d$ is $t_k$ except possibly for the last segment */
  $j \leftarrow t_d - \lfloor \frac{R_k}{\alpha_k}\rfloor$\;
  $F \leftarrow S+t_d-1$\;
  \If {$(j>0)$} {$x[i] \leftarrow j^{th}$ smallest element in $D[S \cdots F]$} \;
   \Else {$x[i] \leftarrow 0$}\;
   \textbf{reserve} $x[i]$ resources at stage $S$.\;
  $T_r \leftarrow T_r-t_d$\;
  $S \leftarrow S+t_d$\;
  $i \leftarrow i+1$\;
    }
 }
 \caption{Single Contract Reservation\label{Algo_SingleContractReservationAlgo}}
\end{algorithm}

For each of the $\lceil \frac{T}{t_k} \rceil$ segments, we find
the $j^{th}$ smallest element of the segment of demand vector
that is covered by that contract where $j=t_d - \lfloor
\frac{R_k}{\alpha_k}\rfloor$~($t_d$ is $t_k$ except possibly for
the last segment). The decided amount of resources are reserved at
the beginning of that segment.

\subsection{Multiple Contracts Reservation}
While Single Contract Reservation algorithm
considers single type of contract, the reservation strategy presented here
considers multiple types of contracts provided by the same CSP as
well as different CSPs. Without loss of generality, we can assume
$t_1 > t_2 > \ldots > t_K$ where $t_k$ is the contract duration of
$k^{th}$ type of reservation contract. We make a realistic
assumption that $R_1/t_1 < R_2/t_2 < \ldots < R_K/t_K$ where
$R_k$ is the reservation charge of $k^{th}$ type of reservation
contract. If for example reservation cost of a one-year contract
is not greater than  one third of that of a three-year contract
no one will ask for a three year contract.
 \begin{algorithm}
 \DontPrintSemicolon
 \SetKwInOut{Input}{input}\SetKwInOut{Output}{output}
 \Input{Demand Vector $D[1..T]$, Contract Vector $C[1..K]$.}
 \Output{Reservation Decision in vector $x[1..K][1..max(\lceil \frac{T}{t_k} \rceil)]$}
 /*$C[i]=(R_i,t_i,\alpha_i)$ and $t_1>t_2>\ldots t_K$ */\\
 \Begin{

 \For{$i\leftarrow 1$ \KwTo $K$}
 {

  $temp \longleftarrow Single-Contract-Reservation(D,C[i])$\;
  /* $temp$ holds the reservation decision vector */\;
  \For{$j\leftarrow 1$ \KwTo $|temp|$ }
  {
    \emph{\small{/*~Update the Remaining Demands~ */}}\;
    \For{$l\leftarrow 1$ \KwTo $t_i$}
    {
    $D[(j-1).t_i+l] \longleftarrow \max (D[(j-1).t_i+l] - temp[j], 0)$\;
    }
  }
 }
 }
 \caption{Multiple Contracts Reservation}
 \label{Multiple Contracts Reservation}
\end{algorithm}

As the reservation cost per hour is minimum for contract type 1,
we would reserve as much as possible for the longest duration
$t_1$ using Single Contract Reservation algorithm for contract
type 1. We will repeat for the successive smaller duration
contracts after updating the remaining demands. After considering
all the available contracts, the remaining demands will be
fulfilled from the on-demand resources.

\section{Implementation and Evaluation}\label{sec:result}
The algorithms for heuristic-based resource reservation have been implemented and their performances have been evaluated and compared with the optimal (IPP-based) one. On-demand and reserved pricing models offered by Amazon EC2 have been considered for this purpose. For reserved pricing model, both, 1-year and 3-years contracts have been considered. The demand vector of VMs have been parsed with the real traces of loads from four different \textsl{standard workload archives} of SDSC Blue Horizon, SDSC SP2, Sandia National Labs and High Performance Computing Center North(Sweden).
\subsection{Experimental Setup}
The IP formulation for this problem has been solved using
\textsl{GLPK (GNU Linear Programming Kit) \cite{glpk} Version
4.47}. The heuristic-based resource reservation algorithms have
been implemented using \textsl{Java SDK v6}. Both the IPP and the
proposed algorithms have been executed on a quad core Intel i5
2.4GHz processor with 4GB RAM.

As discussed earlier, IPP is NP-hard and in our
experimental set-up it had not been possible to
solve problems involving demand vector of size more than 6
months. Therefore, the Amazon EC2 contracts for 1-year and
3-years have been scaled down to 1-month and 3-months
respectively. The associated costs have been scaled down using
the following equations.
\begin{equation}\label{Equ_Scale1}
  R_{1-month} = R_{1-year}/12
\end{equation}
\begin{equation}\label{Equ_scale2}
 R_{3-months} = (R_{3-years}/36)*3
\end{equation}
It is easy to observe that the above two equations keep the hourly discount of reserved VM over on-demand VM unchanged.

We have used a single type of VM in our simulation. However, it can handle multiple types of VM also. For multiple types of VM, we need to determine the demands for each category of VMs separately and then we need to apply our heuristics to find the amount of reservation for each category. The properties of the VM we use in our simulation are listed in Table ~\ref{table_vm_properties}.
\begin{table}[h!]
\renewcommand{\arraystretch}{1.3}
\caption{Properties of the VM used in the Simulation}
\label{table_vm_properties} \centering
\begin{tabular}{|l|c|}
  \hline
  Property & Value \\
  \hline
  \hline
  Type of VM & Standard large (linux) \\
  \hline
  Reservation Cost(1-month) & \$32.00 \\
  \hline
  Reservation Cost(3-months) & \$20.25 \\
  \hline
  On-demand Usage Cost & \$0.24/hour \\
  \hline
  Reserved VM (1-month) Usage Cost & \$0.136/hour \\
  \hline
  Reserved VM (3-month) Usage Cost & \$0.108/hour \\
  \hline
\end{tabular}
\end{table}

\subsection{Experimental Results and Discussion}
Though the experiment has been carried out for four different sets
of demand data, in this section we use the results of SDSC Blue
Horizon and Sandia National Labs workload archives for the
evaluation. The demand data derived from SDSC Blue Horizon
workload archive has the property of uniform distribution of the
demands as shown in Fig.~\ref{Fig_DD_Blue_Horizon} . Whereas a
non-uniform distribution of the demand is observed in the demand
data derived from Sandia National Labs workload archive as shown
in Fig.~\ref{Fig_DD_Sandia}.

\begin{figure}[h!]
\centering
\includegraphics[width=3.7in]{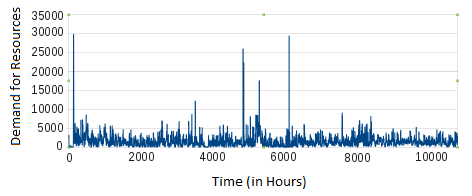}
\caption{Uniform Demand Data of SDSC Blue Horizon}
\label{Fig_DD_Blue_Horizon}
\end{figure}

\begin{figure}[h!]
\centering
\includegraphics[width=3.5in]{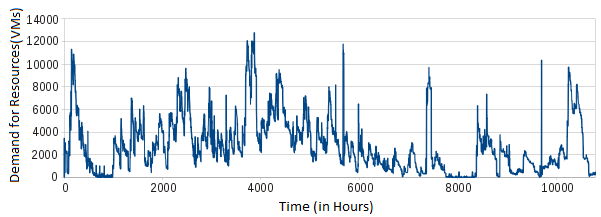}
\caption{Non-Uniform Demand Data of Sandia National Lab}
\label{Fig_DD_Sandia}
\end{figure}

Minimizing the total cost of deploying and running an application
in cloud is the primary objective of optimal resource reservation
problem. So we use \textit{Total Cost} and \textit{Cost Increase}
as the metrics for comparing the results of different reservation
strategies. The metric \textit{Total Cost} includes the one time
reservation cost and the usage costs of on-demand VMs as well as
reserved VMs to meet the demands for a particular duration.
Whereas the metric \textit{Cost Increase} measures the percentage
increase in the total cost in comparison with the Optimal (IPP)
strategy to meet the demands of VMs for a particular duration.
Apart from Total Cost and Cost Increase, we also use the metric
\textit{Overhead} to show the total computational
time needed by different reservation strategies under the
experimental environment. This metric is useful to analyse the
feasibility of a particular reservation strategy.

\begin{figure}[h!]
\centering
\includegraphics[width=3.5in]{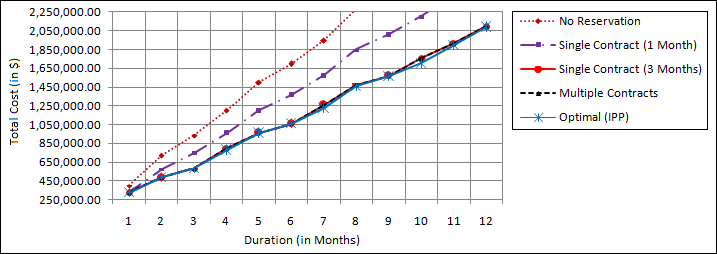}
\caption{Total Cost for SDSC Blue Horizon}
\label{Fig_TotalCost_SDSC}
\end{figure}

\begin{figure}[h!]
\centering
\includegraphics[width=3.5in]{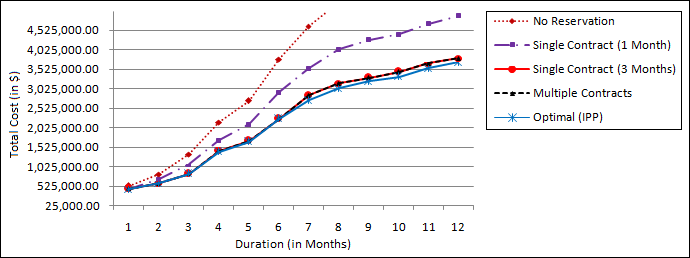}
\caption{Total Cost for Sandia National Lab}
\label{Fig_TotalCost_Sandia}
\end{figure}

Total costs for different duration, ranging from 1 month to 12
months, are shown in Fig.~\ref{Fig_TotalCost_SDSC} for the SDSC
Blue Horizon workload archive and in
Fig.~\ref{Fig_TotalCost_Sandia} for the Sandia National Labs
workload archive. We have compared the total costs for optimal
strategy (IPP), heuristic-based strategies with multiple
contracts, with single contract for 1 month and with single
contract for 3 months. Total cost without any reservation is also
depicted in the figures.

From the above figures, following observations are made:
\begin{enumerate}
  \item While comparing with {\em no reservation}, it has been observed that all the reservation strategies reduce the total cost significantly.
  \item Contracts with longer duration provide reduced costs compared to the contracts with shorter duration. For examople, heuristic-based reservation algorithm for single contract of 3-months outperforms single contract of 1-month reservation. This is because, for any two contracts $i$ and $j$ such that $t_i < t_j$, ${R_{j} \over R_{i}}$ is always less than ${t_{j} \over t_{i}}$.
  \item Total costs of heuristic-based reservation for multiple contracts and single contract for 3 months are very close to the total costs of optimal (IPP) reservation strategy.
 \item The relative performance of different reservation strategies are unaffected by the distribution (uniform or non-uniform) of the demand data.
\end{enumerate}

\begin{table}[h!]
\renewcommand{\arraystretch}{1.3}
\caption{Cost Increase (in \%) of Different Resrvation Strategies compared to Optimal(IPP)Strategy}
\label{table_CostIncrease} \centering
\begin{tabular}{|c|c|c|c|c|c|c|}
  \hline
  Duration~(months) & 1 & 2 & 3 & 4 & 5 & 6 \\
  \hline
  \hline
  No Reservation & 23.61 & 47.88 & 59.95 & 53.50 & 56.00 & 60.05 \\
  \hline
  Single Contract & 1.59 & 19.25 & 28.30 & 23.32 & 25.25 & 29.23 \\
  (1-month) & & & & & & \\
  \hline
  Single Contract & 1.14 & 0.00 & 0.00 & 1.90 & 0.10 & 0.00 \\
  (3-months) & & & & & & \\
  \hline
  Multiple Contracts & 0.93 & 0.01 & 0.00 & 1.82 & 0.10 & 0.00 \\
  \hline
  \hline
  \hline
  Duration~(months) & 7 & 8 & 9 & 10 & 11 & 12 \\
  \hline
  \hline
  No Reservation & 58.71 & 56.00 & 57.75 & 59.18 & 58.15 & 59.62 \\
  \hline
  Single Contract & 28.20 & 26.82 & 28.58 & 29.40 & 28.27 & 29.11 \\
  (1-month) & & & & & & \\
  \hline
  Single Contract & 2.78 & 0.92 & 0.26 & 3.18 & 1.22 & 0.36 \\
  (3-months) & & & & & & \\
  \hline
  Multiple Contracts & 2.74 & 0.92 & 0.26 & 3.15 & 1.22 & 0.36 \\
  \hline
\end{tabular}
\end{table}

The percentage increases in the total costs in case of different heuristic-based strategies over the optimal~(IPP) strategy are tabulated in Table~\ref{table_CostIncrease} for SDSC workload archive. Clearly, Single Contract~(1 month) strategy is not acceptable. However, Single Contract~(3 months) as well as Multiple Contracts perform well and the average increase in the total costs is below 1\%. The relative increases in costs for both, uniform and non-uniform demand data, are shown in Fig.~\ref{Fig_CostIncrease}. Following observations are made from the graph:

\begin{enumerate}
  \item Increase in cost is higher for non-uniform demand data compared to uniform demand data.
  \item For longer duration, the cost increase is higher due to cumulative effect.
  \item The cost increase tends to 0\% for durations which are multiple of 3.
  \item The cost increase is maximum for durations which are multiple of 3 + 1~(e.g. 1, 4, 7 etc) and gradually decreases till the next duration which is multiple of 3.
\end{enumerate}

\begin{figure}[h!]
\centering
\includegraphics[width=3.5in]{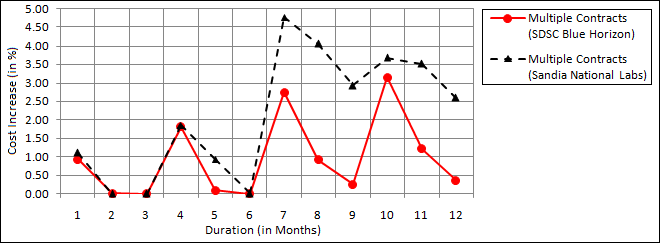}
\caption{Percentage of Increase in Total Cost Relative to Optimal (IPP)}
\label{Fig_CostIncrease}
\end{figure}

From the above analysis, it may be concluded that the heuristic-based reservation for multiple contracts or single contract with longer duration~(in this case 3 months) may be used up to certain demand duration (in this case 6 months -  which actually will be 6 years if scaling down effect is removed).  If demand duration is longer (here longer than 6 months), the cost increase in comparison with the optimal strategy is much higher, and hence the effectiveness of the heristic-based solution is reduced.

Finally, Table~\ref{table_ExecutionTime} shows the effectiveness
of our proposed heuristic-based reservation strategies in terms
of the metric Overhead. While Single Contract (3 months) and
Multiple Contracts require 1113$\mu$s and 2220$\mu$s respectively
for 12 months demand duration, the Optimal strategy requires
7x$10^5 \mu$s for only 1 month demand duration. As expected and
analyzed in Section~\ref{sec_heuristic},
Table~\ref{table_ExecutionTime} shows that the heuristic-based
reservation strategies have a linear time complexity compared to
the exponential time complexity of the optimal strategy.

\begin{table}[h!]
\renewcommand{\arraystretch}{1.3}
\caption{Overhead (in $\mu$s) of Different Reservation Strategies}
\label{table_ExecutionTime} \centering
\begin{tabular}{|c|c|c|c|c|c|c|}
  \hline
  Duration~(months) & 1 & 2 & 3 & 4 & 5 & 6 \\
  \hline
  \hline
  Single Contract & 271 & 463 & 535 & 582 & 659 & 710 \\
  (3-months) & & & & & & \\
  \hline
  Multiple Contracts & 613 & 913 & 1011 & 1170 & 1303 & 1468 \\
  \hline
  Optimal~(IPP) & 7 x & 25 x & 51 x & 101 x & 218 x & 367 x \\
   & $10^5$ & $10^5$ & $10^5$ & $10^5$ & $10^5$ & $10^5$ \\
  \hline  \hline \hline
  Duration~(months) & 7 & 8 & 9 & 10 & 11 & 12 \\
  \hline
  \hline
  Single Contract & 793 & 863 & 924 & 1024 & 1083 & 1113 \\
  (3-months) & & & & & & \\
  \hline
  Multiple Contracts & 1615 & 1726 & 1861 & 1988 & 2130 & 2220 \\
  \hline
  Optimal~(IPP) & 635 & 776 & 1005 & 1228 & 1654 & 1965 \\
   & x$10^5$ & x$10^5$ & x$10^5$ & x$10^5$ & x$10^5$ & x$10^5$ \\
  \hline
\end{tabular}
\end{table}

\section{Conclusion and Future Work} \label{sec:conclude}
In this paper, we have proposed a
heuristic-based algorithm to solve the cost optimization problem
of cloud resource provisioning in linear time. We have  mathematically proved that the
heuristic provides optimal solution under some restrictions. We then extend the heuristic to cover the situations
where these restrictions are removed. The experimental
results show that the solution is very close to the optimal one
with minimal overhead. Even though the IPP based approach gives
the best solution for the problem, the total number of variables becomes
too large if demand data is of duration more than six months and thus become
impossible to solve by existing softwares like {\it GLPK}. At the same time, the linear
time heuristic makes it possible to work on hourly demand data of length 3 years or more without
any difficulty.  We can conclude that the proposed heuristic based approach help to overcome the drawback of
IPP based method  sacrificing very little on the total cost incurred.

To simplify the heuristic, we have ignored the
uncertainty of demands and prices in this work.   Future  work in this area
could be how to bring those uncertainties into consideration.
We would also like to include  the option of spot pricing scheme to  minimize
the total cost further.

\bibliographystyle{IEEEtran}


\end{document}